\documentclass[12pt]{article} 

\usepackage[fleqn]{amsmath}
\usepackage{amsfonts}
\usepackage{amssymb}
\usepackage{latexsym}
\usepackage{xcolor}
\usepackage{graphicx}
\usepackage{float}
\usepackage{subcaption}

\newcommand{\set}[1]{\left\{ #1 \right\}}
\newcommand{\paren}[1]{\left( #1 \right)}
\newcommand{\sqparen}[1]{\left[ #1 \right]}
\newcommand{\abs}[1]{\left\vert #1 \right\vert}

\newcommand{\Ex}[2][]{\mathbb{E}_{#1\!}\sqparen{#2}}
\newcommand{\Var}[1]{\operatorname{Var}\sqparen{#1}}
\newcommand{\Cov}[2]{\operatorname{Cov}\sqparen{#1, #2}}

\newcommand{\Unif}[2]{\operatorname{U}\!\paren{#1, #2}}

\newcommand{\Normal}[3][]{
\ifx\hfuzz#1\hfuzz 
\mathcal{N}\paren{#2, #3}
\else
\mathcal{N}\paren{#1; #2, #3}
\fi
}

\newtheorem{theorem}{Theorem}
\newtheorem{lemma}[theorem]{Lemma}

\newtheorem{thm}{Theorem}
\newtheorem{algorithm2}{Algorithm}

\newcounter{problemC}[section]

\usepackage{multirow}
\usepackage{epsf}

\begin{document}

\newcommand*{\QEDA}{\hfill\ensuremath{\blacksquare}}

\title{Cardinality Estimation Meets Good-Turing}
\author{
Reuven Cohen~~~Liran Katzir~~~Aviv Yehezkel\\
Department of Computer Science\\
Technion\\
Haifa 32000, Israel\\
}
\date{\today}
\maketitle

\begin{abstract}

Cardinality estimation algorithms receive a stream of elements whose order might be arbitrary, with possible repetitions, and return the number of distinct elements. Such algorithms usually seek to minimize the required storage and processing at the price of inaccuracy in their output. 
Real-world applications of these algorithms are required to process large volumes of monitored data, making it impractical to collect and analyze the entire input stream. In such cases, it is common practice to sample and process only a small part of the stream elements.
This paper presents and analyzes a generic algorithm for combining every cardinality estimation algorithm with a sampling process. We show that the proposed sampling algorithm does not affect the estimator's asymptotic unbiasedness, and we analyze the sampling effect on the estimator's variance.

\end{abstract}

\section{Introduction} \label{sec:intro}

Consider a very long stream of elements $x_1, x_2, x_3, \ldots,$ with repetitions.
Finding the number $n$ of distinct elements is a well-known problem with numerous applications. The elements might represent IP addresses of packets passing through a router~\cite{Giroire2007,Ganguly2007,Metwally:2008}, elements in a large database~\cite{Heule2013}, motifs in a DNA sequence~\cite{Giroire2006}, or nodes of RFID/sensor networks~\cite{Qian2011}. 
One can easily find the exact value of $n$ by comparing the value of a newly encountered element, $x_i$, to every (stored) value encountered so far. If the value of $x_i$ has not been seen before, it is stored as well. After all of the elements are treated, the stored elements are counted. This simple approach does not scale if storage is limited, or if the computation performed for each element $x_i$ should be minimized.
In these cases, the following cardinality estimation problem should be solved:
\refstepcounter{problemC}
\begin{description}\label{prob:cardinality}
\item[The cardinality estimation problem]
\item[Instance:] A stream of elements $x_1, x_2, x_3, \ldots$ with repetitions, and an integer $m$. Let $n$ be the number of different elements, namely $n=\abs{\set{x_1, x_2, x_3, \ldots}}$, and let these elements be $\set{e_1,e_2,\ldots,e_n}$.
\item[Objective:] Find an estimate $\widehat{n}$ of $n$ using only $ m $ storage units, where $m \ll n$.
\end{description}

As an application example, $x_1,x_2,x_3,\ldots$ could be IP packets received by a server. Each packet belongs to one of $n$ IP flows $e_1,e_2,\ldots,e_n$, and the cardinality $n$ represents the number of active flows. 
By monitoring the number of distinct flows during every time period, a router can estimate the network load imposed on the end server and detect anomalies.
For example, it can detect DDoS attacks on the server when the number of flows significantly increases during a short time interval \cite{Estan06,Ganguly2007}.

Several algorithms have been proposed for the cardinality estimation problem \cite{Cosma:2011, CohenK08, Flajolet2007, Giroire2009, Lumbroso2010, Metwally:2008}, all of which were designed to work on the entire stream, namely, without sampling. However, real-world applications are required to process large volumes of monitored data, making it impractical to collect and process the entire stream. For example, this is the case for IP packets received over a high-speed link, because a 100 Gbps link creates a 1 TB log file in less than 1.5 minutes.
In such cases, only a small part of the stream is sampled and processed \cite{Duffield04,Duffield01}.

In this paper we present and analyze a generic algorithm that adds a sampling process into every cardinality estimation procedure. 
The proposed algorithm consists of two steps: (a) cardinality estimation of the sampled stream using any known cardinality estimator; (b) estimation of the sampling ratio. We show that the proposed algorithm does not affect the original estimator's asymptotic bias (accuracy), and we analyze the algorithm's effect on the estimator's variance (precision).

A naive approach to solving the cardinality estimation problem is to estimate the cardinality of the sampled stream and view it as an estimation for the cardinality of the whole (unsampled) stream. However, this approach yields poor results because it ignores the probability of elements that do not appear in the sample.
For example, we simulated a stream of $n=10,000$ distinct elements whose frequency in the stream follows uniform distribution $\sim \Unif{10^2}{10^4}$. We then sampled $0.1 \%$ of the stream and used the HyperLogLog algorithm \cite{Flajolet2007} with $m=200$ storage units to estimate the cardinality of the sample. We repeated this test $200$ times, each on a different stream of $10,000$ distinct elements, and averaged the results. We found that the mean estimated cardinality is $\Ex{\widehat{n}} \approx 9,100 $, which means a bias of $9 \%$, and that the relative variance is $\Var{\frac{\widehat{n}}{n}} \approx 0.0552 $. In contrast, our proposed algorithm computed a mean estimated cardinality of $\Ex{\widehat{n}} \approx 9,900 $, namely a bias of only $1 \%$, and a relative variance of only $\Var{\frac{\widehat{n}}{n}} \approx 0.0118 $.

The rest of this paper is organized as follows. Section \ref{sec:related} discusses previous work. Section \ref{sec:algorithm} presents our first algorithm (Algorithm \ref{alg:cardEstWithSampling}) for combining the sampling process with a generic cardinality estimation procedure. In addition, this section presents an analysis of the asymptotic bias and variance of Algorithm \ref{alg:cardEstWithSampling}. Section \ref{sec:newAlg} presents our enhanced algorithm (Algorithm \ref{alg:cardEstWithSubSampling}), which uses subsampling in order to reduce the memory cost of Algorithm \ref{alg:cardEstWithSampling}. This section also presents an analysis of the asymptotic bias and variance of Algorithm \ref{alg:cardEstWithSubSampling}. 
Section \ref{sec:sim} presents simulation results that validate our analysis in Sections \ref{sec:algorithm} and \ref{sec:newAlg}. 
Finally, Section \ref{sec:conclusion} concludes the paper.

\section{Related Work} \label{sec:related}

Several works address the cardinality estimation problem \cite{Cosma:2011, CohenK08, Flajolet2007, Giroire2009, Lumbroso2010, Metwally:2008} and propose statistical algorithms for solving it. These algorithms are efficient because they make only one pass on the data stream, and because they use a fixed and small amount of storage.
The common approach is to use a random hash function that maps each element $e_j$ into a low-dimensional data sketch $h(e_j)$, which can be viewed as a random variable.
The hash function guarantees that $h(e_j)$ is identical for all the appearances of $e_j$. Thus, the existence of duplicates, i.e., multiple appearances of the same element, does not affect the value of the extreme order statistics.
Let $h$ be a hash function and $h(x_i)$ denote the hash value of $x_i$. Then, an order statistics estimator or a bit pattern estimator can be used to estimate the value of $n$. An order statistics estimator keeps the smallest (or largest) $ m $ hash values. These values are then used to estimate the cardinality \cite{Chassaing2006, CohenK08, Giroire2009, BarYossef2002, Lumbroso2010}. A bit pattern estimator keeps the highest position of the leftmost (or rightmost) ``1'' bit in the binary representation of the hash values in order to estimate the cardinality \cite{Cosma:2011, Flajolet2007}.

Real-world applications of cardinality estimation algorithms are required to process large volumes of monitored data, making it impractical to collect and analyze the entire input stream. In such cases, it is common practice to sample and process only a small part of the stream elements. For example, routers use sampling techniques to achieve scalability. The industry standard for packet sampling is sFlow \cite{sFlow1}, short for ``sampled flow". Using a defined sampling rate $N$, an average of 1 out of $N$ packets is randomly sampled. The flow samples are then sent as sFlow datagrams to a central monitoring server, which analyzes the network traffic. 

Although sampling techniques provide greater scalability, they also make it more difficult to infer the characteristics of the original stream. One of the first works addressing inference from samples is the Good-Turing frequency estimation, a statistical technique for estimating the probability of encountering a hitherto unseen element in a stream, given a set of past samples. For a recent paper on the Good-Turing technique, see \cite{Gale95}.

Several other works have addressed the problem of inference from samples.
For example, the detection of heavy hitters, elements that appear many times in the stream, is studied in \cite{Bhatt07}. The authors propose to keep track of the volume of data that has not been sampled. Then, a new element is skipped only when its effect on the estimation will ``not be too large."
The case where the elements are packets has also been addressed. In such cases, the heavy hitters are called elephants. 
The accuracy of detecting elephant flows is studied in \cite{Mori07} and \cite{Mori04}. The authors use Bayes' theorem for determining the threshold of sampled packets, which indicates whether or not a flow is an elephant in the entire stream.

Other works have dealt with exploiting protocol-level information of sampled packets in order to obtain accurate estimations of the size of flows in the network. 
For example, in \cite{Duffield03} the authors present a TCP-specific method whose estimate is based on the TCP SYN flag in the sampled packets.
Another method, which uses TCP sequence numbers, is presented in \cite{Ribeiro06}.
These methods can also be used to estimate the cardinality of the flows in the network, i.e., the number of active flows. However, both methods are limited to TCP flows. In this paper we present a generic algorithm that does not make any assumptions regarding the type of the input elements. 

Related to the cardinality estimation problem is the problem of finding a uniform sample of the distinct values in the stream. Such a sample can be used for a variety of database management applications, such as query optimization, query monitoring, query progress indication and query execution time prediction \cite{Babcock03,Chaudhuri01,Chaudhuri07}. Additional applications of the uniform sample pertain to approximate query answering, such as estimating the mean, the variance, and the quantiles over the distinct values of the query \cite{Acharya00,Acharya99,Gibbons98}.
Several algorithms provide a uniform sample of the stream; for example, the authors of \cite{Gibbons01} show how to find such a sample in a single data pass. Several variations of this work are also proposed in \cite{Cormode05,Frahling08,Ganguly07}. However, all the discussed approaches require scanning the entire input stream, which is usually impractical. In this paper we present a generic algorithm that does not require a full data pass over the input stream.

The above works consider uniform packet sampling, where each packet is sampled with a fixed probability. 
Previous works have also dealt with size-dependent flow sampling, where packets are sampled with different probability, according to their flow size. 
The first works on size-dependent flow sampling study the problem of deciding which records in a given set of flow records should be discarded when storage constraints allow only a small fraction to be kept \cite{DuffieldL03,Duffield01,DuffieldLT04}. 
The sampling decision in these works is made off-line: a flow is first received and only then discarded or stored. In \cite{Kumar06}, the on-line version of this problem is studied. In this version, upon receiving a packet, the algorithm needs to determine whether to keep it. The authors develop a new packet sampling method that samples each packet with probability $f(\widehat{s})$, where $f$ is a decreasing function of the estimated size of the corresponding flow when the packet is received, and the size of the flow is estimated using a small sketch that stores the approximate sizes of all flows.

\section{Cardinality Estimation with Sampling}\label{sec:algorithm}

\subsection{Preliminaries: Good-Turing Frequency Estimation} \label{sub:pre}

The Good-Turing frequency estimation technique is useful in many language-related tasks where one needs to determine the probability that a word will appear in a document. 

Let $X=\set{x_1,x_2,x_3,\ldots}$ be a stream of elements, and let $E$ be the set of all different elements $E=\set{e_1,e_2,\ldots,e_n}$, such that $x_i \in E$. 
Suppose that we want to estimate the probability $\pi(e_j)$ that a randomly chosen element from $X$ is $e_j$. A naive approach is to choose a sample $Y = \set{y_1,y_2,\ldots,y_l}$ of $l$ elements from $X$, and then to let $\pi(e_j) = \frac{\#(e_j)}{l}$, where $\#(e_j)$ denotes the number of appearances of $e_j$ in $Y$. However, this approach is inaccurate, because for each element $e_j$ that does not appear in $Y$ even once (an ``unseen element"), $\#(e_j)=0$, and therefore $\pi(e_j)=0$. 

Let $E_i=\set{e_j | \#(e_j)=i}$ be the set of elements that appear $i$ times in the sample $Y$. Thus, $\sum{\abs{E_i} \cdot i}=l$. The Good-Turing frequency estimation claims that $\widehat{P_i}=(i+1)\frac{\abs{E_{i+1}}}{l}$ is a consistent estimator for the probability $P_i$ that an element of $X$ appears in the sample $i$ times.

For the special case of $P_0$, we get from Good-Turing that $\widehat{P_0}=\abs{E_1}/l$. In other words, the hidden mass $P_0$ can be estimated by the relative frequency of the elements that appear exactly once in the sample $Y$. For example, if $1/10$ of the elements in $Y$ appear only once in $Y$, then approximately $1/10$ of the elements in $X$ do not appear in $Y$ at all (i.e., they are unseen elements).

\subsection{The Proposed Algorithm} \label{sub:algo}

We now show how to use Good-Turing in order to combine a sampling process with a generic cardinality estimation procedure, referred to as Procedure 1. 
As before, let $X=\set{x_1,x_2,x_3,\ldots}$ be the entire stream of elements, and let $Y=\set{y_1, y_2, \ldots, y_l}$ be the sampled stream.
Assume that the sampling rate is $P$, namely, $1/P$ of the elements of $X$ are sampled into $Y$. Let $n$ and $n_s$ be the number of distinct elements in $X$ and $Y$ respectively. The algorithm receives the sampled stream $Y$ as an input and returns an estimate for $n$.
The algorithm consists of two steps: (a) estimating $n_s$ using Procedure 1 (any procedure, such as in \cite{Cosma:2011, Flajolet2007, Giroire2009, Lumbroso2010}); (b) estimating $n/{n_s}$, the factor by which to multiply the cardinality $n_s$ of the sampled stream in order to estimate the cardinality $n$ of the full stream.

To estimate $n_s$ in step (a), Procedure 1 is invoked using $m$ storage units. To estimate $n/{n_s}$ in step (b), we note that $P_0 = (n-n_s)/n$ and thus $1/(1-P_0) = n/{n_s}$. Therefore, the problem of estimating $n/{n_s}$ is reduced to estimating the probability $P_0$ of unseen elements.
As indicated above, by Good-Turing, $\widehat{P_0}=\abs{E_1}/l$ is a consistent estimator for $P_0$. Thus, we only need to find the number $\abs{E_1}$ of elements that appear exactly once in the sampled stream $Y$. To compute the value of $\abs{E_1}$ precisely, one should keep track of all the elements in $Y$ and ignore each previously encountered element. This is done by Algorithm \ref{alg:cardEstWithSampling} below using $O(l)$ storage units. We later show (Algorithm \ref{alg:cardEstWithSubSampling} in Section \ref{sec:newAlg}) that the number of storage units can be reduced by estimating the value of $\abs{E_1}/l$.
\begin{algorithm2}
\textbf{\newline (cardinality estimation with sampling)}
\label{alg:cardEstWithSampling}
\begin{enumerate}
\item [(a)] Estimate the number $n_s$ of distinct elements in the sample $Y$ by invoking a cardinality estimation procedure (Procedure 1) on this sample using $m$ storage units.
\item [(b)] Determine the ratio $n/{n_s}$ by computing $\frac{1}{1-\widehat{P_0}}$, where $\widehat{P_0} = \abs{E_1}/l$. The value of $\abs{E_1}$ is computed precisely and $l$ is known.
\item [(c)] Return $\widehat{n} = \widehat{n_s} \cdot \widehat{n/{n_s}}$ as an estimator for the cardinality of the entire stream $X$.
\end{enumerate}
\end{algorithm2}

\subsection{Analysis of Algorithm \ref{alg:cardEstWithSampling}} \label{sub:analysis}

In this section we analyze the asymptotic bias and variance of Algorithm \ref{alg:cardEstWithSampling}, assuming that the HyperLogLog algorithm \cite{Flajolet2007} is used as Procedure 1.
This algorithm is the best known cardinality estimator and it has a relative variance of $ \Var{\frac{\widehat{n}}{n}} \approx 1.08/m $, where $ m $ is the number of used storage units.
Our main result is Theorem \ref{thm1}, where we prove that the sampling does not affect the estimator's asymptotic unbiasedness, and we show the effect of the sampling rate $P$ on the estimator's variance.

We start with three preliminary lemmas. The first lemma shows how to compute the probability distribution of a random variable that is a product of two normally distributed random variables whose covariance is $0$:
\begin{lemma}[Product distribution]\label{lemma:productDis} \ \\
Let $X$ and $Y$ be two random variables satisfying $X \to \Normal{\mu_x}{\sigma_x^2}$ and $Y \to \Normal{\mu_y}{\sigma_y^2}$, such that $\Cov{X}{Y}=0$.
Then, the product $X \cdot Y$ asymptotically satisfies the following:
\begin{equation*}
X\cdot Y \to \Normal{\mu_x\mu_y}{\mu_y^2\sigma_x^2 + \mu_x^2\sigma_y^2} \text{.}
\end{equation*}
\end{lemma} 
A proof is given in \cite{Shao1998}.

The next lemma, known as the Delta Method, can be used to compute the probability distribution for a function of an asymptotically normal estimator using the estimator's variance:
\begin{lemma}[Delta Method] \label{lemma:deltaMethod} \ \\
Let $\theta_m$ be sequence of random variables satisfying $\sqrt{m}(\theta_m - \theta) \to \Normal{0}{\sigma^2}$, where $\theta$ and $\sigma^2$ are finite valued constants. Then, for every function $g$ for which $g^\prime(\theta)$ exists and $g^\prime(\theta) \neq 0$, the following holds:
\begin{equation*}
\sqrt{m}(g(\theta_m) - g(\theta)) \to \Normal{0}{\sigma^2 {g^\prime(\theta)}^2}\text{.}
\end{equation*}
\end{lemma} 
A proof is given in \cite{Shao1998}.

The last lemma states a normal limit law for the estimation of $\abs{E_1}/l$, where $\abs{E_1}$ and $l$ are as described in Section \ref{sub:pre}:
\begin{lemma}[Random Sample's Coverage] \label{lemma:N1acc} \ \\
$\widehat{\abs{E_1}/l} \to \Normal{\abs{E_1}/l}{\frac{1}{l}((\abs{E_1}+2\abs{E_2})/l - (\abs{E_1}/l)^2)}$.
\end{lemma}
A proof is given in \cite{Esty83}.

We are now ready to start our analysis. Our first lemma summarizes the distribution of $P_0$:
\begin{lemma} \label{lemma:p0acc} \ \\
$\widehat{P_0} \to \Normal{P_0}{\frac{1}{l}\Big(P_0(1-P_0) + P_1\Big)}$, where $l$ is the sample size.
\end{lemma}
\begin{proof} \ \\
For the expectation, the following holds
\begin{align*}
\Ex{\widehat{P_0}} = \Ex{\abs{E_1}/l} = \abs{E_1}/l
\text{.}
\end{align*}
The first equality is due to the definition of $\widehat{P_0}$ in Algorithm \ref{alg:cardEstWithSampling}, and the second is because $\abs{E_1}$ and $l$ are constants. By Good-Turing we get that $\abs{E_1}/l \to P_0$.

For the variance, the following holds
\begin{align*}
\Var{\widehat{P_0}} = \Var{\abs{E_1}/l} = 1/l \cdot ((\abs{E_1}+2\abs{E_2})/l - (\abs{E_1}/l)^2) 
\text{.}
\end{align*}
The first equality is due to the definition of $\widehat{P_0}$ in Algorithm \ref{alg:cardEstWithSampling}. The second equality is due to Lemma \ref{lemma:N1acc}. 
Finally, due to Good-Turing we get that $1/l \cdot ((\abs{E_1}+2\abs{E_2})/l - (\abs{E_1}/l)^2) \to \frac{1}{l}\Big(P_0(1-P_0) + P_1\Big)$.
\end{proof}

As shown in \cite{Flajolet2007}, when sampling is not used, Procedure 1 estimates $n$ with mean value $n$ and variance $\frac{n^2}{m}$, namely, $\widehat{n} \to \Normal{n}{\frac{n^2}{m}}$. The following theorem states the asymptotic bias and variance of Algorithm \ref{alg:cardEstWithSampling} for $P<1$.

\begin{thm} \ \\
\label{thm1}
Algorithm \ref{alg:cardEstWithSampling} estimates $n$ with mean value $n$ and variance $\frac{n^2}{l}\frac{P_0(1-P_0)+P_1}{(1-P_0)^2} + \frac{n^2}{m}$, namely, $\widehat{n} \to \Normal{n}{\frac{n^2}{l}\frac{P_0(1-P_0)+P_1}{(1-P_0)^2} + \frac{n^2}{m}}$, where $l$ is the sample size, and $m$ is the storage size used for estimating $n_s$. 
In addition, $P_0$ and $P_1$ satisfy:
\begin{enumerate}
\item $\Ex{P_0} = \frac{1}{n}\sum_{i=1}^{n}{e^{-P\cdot f_i}}$.
\item $\Ex{P_1} = \frac{P}{n}\sum_{i=1}^{n}{f_i \cdot e^{-P\cdot f_i}}$
\end{enumerate}
where $f_i$ is the frequency of element $e_i$ in $X$.
\end{thm}
\begin{proof} \ \\
Applying the Delta Method (Lemma \ref{lemma:deltaMethod}) on $1-P_0$ yields that
\begin{equation}\label{eq:1minus}
\frac{1}{1-\widehat{P_0}} \to \Normal{\frac{1}{1-P_0}}{\frac{1}{l}\frac{P_0(1-P_0) + P_1}{(1-P_0)^4}} \text{.}
\end{equation}
According to \cite{Cosma:2011}:
\begin{equation}\label{eq:v}
\widehat{n}_s \to \Normal{n_s}{\frac{n_s^2}{m}} \text{.}
\end{equation}
Next we show that $\frac{1}{1-P_0}$ and $n_s$ have zero covariance: 
\begin{align*}
\Cov{n_s}{\frac{1}{1-P_0}} &= \Cov{n_s}{\frac{n}{n_s}} \\
&= \Ex{\Cov{n_s}{\frac{n}{n_s} \mid n_s}} + \Cov{\Ex{n_s \mid n_s}}{\Ex{\frac{n}{n_s} \mid n_s}} \\
&= 0 + \Cov{n_s}{\frac{n}{n_s}} \\
&= \Ex{n_s \cdot \frac{n}{n_s}} - \Ex{n_s} \Ex{\frac{n}{n_s}} \\
&= \Ex{n} - n_s \cdot \frac{n}{n_s} \\
&= n - n = 0 \text{.}
\end{align*}
The first equality is due to the $P_0$ definition. The second equality is due to the law of total covariance. The third equality is because $n_s$ and $n/n_s$ are independent when $n_s$ is known. The fourth equality is due to the covariance definition. The fifth and sixth equalities are due to the expectation definition and algebraic manipulations.

Applying the distribution product property (Lemma \ref{lemma:productDis}) for Eqs. (\ref{eq:1minus}) and (\ref{eq:v}) yields that:
\begin{equation*}\label{eq:total}
\widehat{n} = \frac{\widehat{n}_s}{1-\widehat{P_0}} \to \Normal{n}{\frac{n_s^2}{l}\frac{P_0(1-P_0)+P_1}{(1-P_0)^4} + \frac{n_s^2}{m}\frac{1}{(1-P_0)^2}} \text{.}
\end{equation*}
Finally, substituting $n_s = n \cdot (1-P_0)$ yields that:
\begin{equation*}\label{eq:total}
\widehat{n} \to \Normal{n}{\frac{n^2}{l}\frac{P_0(1-P_0)+P_1}{(1-P_0)^2} + \frac{n^2}{m}} \text{.}
\end{equation*}

The resulting asymptotic variance depends on both $P_0$ and $P_1$, which are determined according to the sampling rate $P$ and $f_i$, the frequency of each distinct element in the stream. Thus, the final part of the proof is to compute their expectation.
For $P_0$ we get that:
\begin{equation*}
\Ex{P_0} = \frac{1}{n}\sum_{i=1}^{n}{(1-P)^{f_i}} = \frac{1}{n}\sum_{i=1}^{n}{((1-P)^{1/P})^{P \cdot f_i}} = \frac{1}{n}\sum_{i=1}^{n}{(e^{-1})^{P \cdot f_i}} = \frac{1}{n}\sum_{i=1}^{n}{e^{-P \cdot f_i}} \text{.}
\end{equation*}
The first equality is due to the expectation and $P_0$ definitions. The second and the last equalities are due to algebraic manipulations. The third equality is due to the known limit result where $(1-x)^{1/x} \to e^{-1}$ when $x \to 0$ (in our case $P \to 0$).

For $P_1$ we get that:
\begin{equation*}
\Ex{P_1} = \frac{1}{n}\sum_{i=1}^{n}{f_i \cdot P(1-P)^{f_i-1}} = \frac{1}{n}\sum_{i=1}^{n}{f_i \cdot P((1-P)^{1/P})^{P \cdot (f_i-1)}} = \frac{P}{n}\sum_{i=1}^{n}{f_i \cdot e^{-P \cdot f_i}} \text{.}
\end{equation*}
The first equality is due to the expectation and $P_1$ definitions. The second and third equalities are due to algebraic manipulations and the same known limit result noted above. 
\end{proof}

\section{Reducing the Computational Cost of Algorithm \ref{alg:cardEstWithSampling}} \label{sec:newAlg}

\subsection{Algorithm \ref{alg:cardEstWithSubSampling} with Subsampling}

Algorithm \ref{alg:cardEstWithSampling} computes $\abs{E_1}$ precisely. To this end, it uses $O(l)$ storage units, which is linear in the sample size. We now show how to reduce this cost by approximating the value of $\abs{E_1}$ using a subsample $U$ of the sample $Y$ (see Figure \ref{fig:relationship}).
\begin{algorithm2}
\textbf{\newline (cardinality estimation with sampling and subsampling)}
\label{alg:cardEstWithSubSampling}\ \\
Same as Algorithm \ref{alg:cardEstWithSampling}, except that in step (b) the ratio $\abs{E_1}/l$ is estimated by invoking Procedure 2 using only $u \ll l$ storage units.
\end{algorithm2}

\begin{figure}[tbp]
\begin{center}
\epsfxsize=0.33\textwidth \epsffile{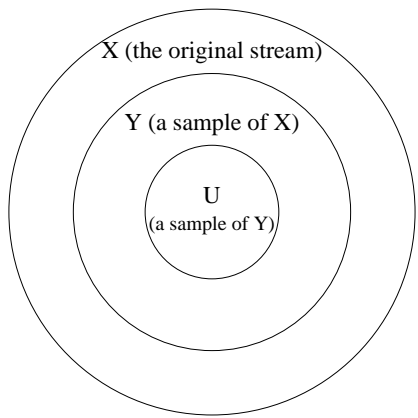}
\caption{The relationship between X, Y and U}
\label{fig:relationship}
\end{center}
\end{figure}

\begin{flushleft}
\textbf{Procedure 2:}
\\
\end{flushleft}
\begin{enumerate}
\item Uniformly subsample $u$ elements from the sampled stream $Y$. Let this subsample be $U$.
\item Compute (precisely) the number $\abs{U_1}$ of elements that appear only once in $U$.
\item Return $\widehat{P_0} = \abs{U_1}/u$.
\end{enumerate}

The intuition behind Algorithm \ref{alg:cardEstWithSubSampling} is that the cheap operation of Algorithm \ref{alg:cardEstWithSampling}, estimating $n_s$, is performed on the whole sample $Y$, whose length is $l$, while the expensive operation, computing the number of elements that appear only once ($\abs{E_1}$), is performed on a small subsample $U$ of length $u$, where $u \ll l$.

Uniform subsampling (step (1) in Procedure 2) can be implemented using one-pass reservoir sampling \cite{Vitter85}, as follows. First, initialize $U$ with the first $u$ elements of $Y$, namely, $y_1,y_2,\ldots,y_u$, and sort them in decreasing order of their hash values. When a new element is sampled into $Y$, its hash value is compared to the current maximal hash value of the elements in $U$. If the hash value of the new element is smaller than the current maximal hash value of $U$, the new value is stored in $U$ instead of the element with the maximal hash value. After all of the elements are treated and the sample $Y$ is created, $U$ is a uniform subsample of length $u$.

We now analyze the running time complexity of Algorithm \ref{alg:cardEstWithSubSampling}. Both steps (a) and (b) are performed using a simple pass over the sample $Y$, and require $O(1)$ operations per sampled element. Thus, these steps require $O(l)$ operations. Step (b) requires additional $O(u)$ operations for each insertion of an element into $U$. On the average, there are $O(\log l)$ such insertions. The total complexity is thus $O(l + u\cdot \log l) = O(l)$, which is similar to that of Algorithm \ref{alg:cardEstWithSampling}.
However, the main advantage of Algorithm \ref{alg:cardEstWithSubSampling} over Algorithm \ref{alg:cardEstWithSampling} is that it requires only $m+u$ storage units, while Algorithm \ref{alg:cardEstWithSampling} requires $m+l$ storage units, where $u \ll l$.

Next, we analyze the asymptotic bias and variance of Algorithm \ref{alg:cardEstWithSubSampling}, assuming that the HyperLogLog algorithm \cite{Flajolet2007} is used as Procedure 1. Then we generalize the analysis for any cardinality estimation procedure.

\subsection{Analysis of Algorithm \ref{alg:cardEstWithSubSampling}} \label{sub:alg2Ana}

Our main result is Theorem \ref{thm12}, which proves that the subsampling does not affect the asymptotic unbiasedness of the estimator and analyzes the effect of the sampling rate $P$ on the estimator's variance, with respect to the storage sizes $m$ and $u$.

Let $Z_i$ be the set of elements that appear exactly $i$ times in the subsample $U$; thus, $\sum \abs{Z_i} \cdot i = u$ and $Z_1$ is the set of elements that appear only once in $U$. $\abs{Z_1}$ can be written using indicator variables as:
\begin{equation*}
\abs{Z_1} = \sum_{j=1}^{u} {I_j}\text{, \:\:\:\:\: where}
\end{equation*}
$
I_j =
\left\{
\begin{array}{ll}
1 & \text{if the j'th element in $U$ has a single appearance in the subsample} \\
0 & \text{otherwise.} 
\end{array}
\right.
$

Consider the estimator $\widehat{\abs{E_1}/l}$ for $\abs{Z_1}/u$.
By definition, the variable $\abs{Z_1}$ follows a hypergeometric distribution, which can be relaxed to a binomial distribution if $u \ll l$ \cite{Shao1998}. Thus, due to binomial distribution properties, the expectation is
\begin{equation} \label{eq:expN12}
\Ex{\widehat{\abs{E_1}/l} \: \mid \: \abs{E_1}} = \Ex{\abs{Z_1}/u} = \Ex{I_j} = \abs{E_1}/l \text{,}
\end{equation}
and the variance is
\begin{equation} \label{eq:varN12}
\Var{\widehat{\abs{E_1}/l} \: \mid \: \abs{E_1}} = \Var{\abs{Z_1}/u} = 1/u \cdot \Var{I_j} = 1/u \cdot \abs{E_1}/l \cdot (1-\abs{E_1}/l) \text{.}
\end{equation} 
The following lemma summarizes the distribution of $P_0$:
\begin{lemma} \label{lemma:p0acc2} \ \\
$\widehat{P_0} \to \Normal{P_0}{\frac{1}{u}\Big(2P_0(1-P_0) + P_1\Big)}$.
\end{lemma}
\begin{proof} \ \\
For the expectation, the following holds
\begin{align*}
\Ex{\widehat{P_0}} = \Ex{\widehat{\abs{E_1}/l}} = \Ex{\Ex{\widehat{\abs{E_1}/l} \: \mid \: \abs{E_1}}} = \Ex{\abs{E_1}/l} = \abs{E_1}/l
\text{.}
\end{align*}
The first equality is due to Procedure 2. The second equality is due to the law of total expectation. The third equality is due to Eq. \ref{eq:expN12}. The fourth equality is due to Lemma \ref{lemma:N1acc}. 

By Good-Turing we get that $\abs{E_1}/l \to P_0$. For the variance, the following holds:
\begin{align*}
\Var{\widehat{P_0}} &= \Var{\widehat{\abs{E_1}/l}} \\ \nonumber 
&= \Var{\Ex{\widehat{\abs{E_1}/l} \: \mid \: \abs{E_1}}} + \Ex{\Var{\widehat{\abs{E_1}/l} \: \mid \: \abs{E_1}}} \\ \nonumber
&= 1/u \cdot ((\abs{E_1}+2\abs{E_2})/l - (\abs{E_1}/l)^2) + 1/u \cdot \abs{E_1}/l \cdot (1-\abs{E_1}/l) \\ \nonumber
&= 2/u \cdot ((\abs{E_1}+\abs{E_2})/l - (\abs{E_1}/l)^2) \nonumber 
\text{.}
\end{align*}
The first equality is due to Procedure 2. The second equality is due to the law of total variance. The third equality is due to Eq. \ref{eq:varN12} and Lemma \ref{lemma:N1acc}. The fourth equality is due to algebraic manipulations. 

By Good-Turing we get that $2/u \cdot ((\abs{E_1}+\abs{E_2})/l - (\abs{E_1}/l)^2) \to \frac{1}{u}\Big(2P_0(1-P_0) + P_1\Big)$.
\end{proof}

The following theorem states the asymptotic bias and variance of Algorithm \ref{alg:cardEstWithSubSampling} for $P<1$.

\begin{thm} \ \\
\label{thm12}
Algorithm \ref{alg:cardEstWithSubSampling} estimates $n$ with mean value $n$ and variance $\frac{n^2}{u}\frac{2P_0(1-P_0)+P_1}{(1-P_0)^2} + \frac{n^2}{m}$, namely, $\widehat{n} \to \Normal{n}{\frac{n^2}{u}\frac{2P_0(1-P_0)+P_1}{(1-P_0)^2} + \frac{n^2}{m}}$. 
In addition, $P_0$ and $P_1$ can be estimated as described in Theorem \ref{thm1}.
\end{thm}
\begin{proof} \ \\
Applying the Delta Method (see Section \ref{sub:analysis}) on $1-P_0$ yields that:
\begin{equation}\label{eq:1minus2}
\frac{1}{1-\widehat{P_0}} \to \Normal{\frac{1}{1-P_0}}{\frac{1}{u}\frac{2P_0(1-P_0) + P_1}{(1-P_0)^4}} \text{.}
\end{equation}
According to \cite{Cosma:2011}:
\begin{equation}\label{eq:v2}
\widehat{n}_s \to \Normal{n_s}{\frac{n_s^2}{m}} \text{.}
\end{equation}
Recall that $\Cov{n_s}{\frac{1}{1-P_0}} = 0$ (see Section \ref{sub:analysis}); applying the distribution product property (see Section \ref{sub:analysis}) for Eqs. (\ref{eq:1minus2}) and (\ref{eq:v2}) yields that: 
\begin{equation*}\label{eq:total2}
\widehat{n} = \frac{\widehat{n}_s}{1-\widehat{P_0}} \to \Normal{n}{\frac{n_s^2}{u}\frac{2P_0(1-P_0)+P_1}{(1-P_0)^4} + \frac{n_s^2}{m}\frac{1}{(1-P_0)^2}} \text{.}
\end{equation*}
Finally, substituting $n_s = n \cdot (1-P_0)$ yields that:
\begin{equation*}\label{eq:total2}
\widehat{n} \to \Normal{n}{\frac{n^2}{u}\frac{2P_0(1-P_0)+P_1}{(1-P_0)^2} + \frac{n^2}{m}} \text{.}
\end{equation*}

The resulting asymptotic variance depends on both $P_0$ and $P_1$, which are determined according to the sampling rate $P$ and $f_i$, the frequency of each distinct element in the stream, as was described in Section \ref{sub:analysis}.
\end{proof}

The analysis above assumes that the HyperLogLog algorithm \cite{Flajolet2007} is used as Procedure 1. Recall that the asymptotic relative efficiency (ARE) of cardinality estimator $\widehat{n}$ is defined as the ratio $\text{ARE} = \frac{n^2}{m} \cdot \frac{1}{\Var{\widehat{n}}}$. For example, the ARE of bottom-$m$ sketches \cite{Giroire2009} is $1.00$, and the ARE of the maximal-term sketch in \cite{Cosma:2011} is $0.93$. The following theorem generalizes Theorem \ref{thm12} for any cardinality estimation procedure.

\begin{thm} \ \\
\label{thm22}
Algorithm \ref{alg:cardEstWithSubSampling} estimates $n$ with mean value $n$ and variance $\frac{n^2}{u}\frac{2P_0(1-P_0)+P_1}{(1-P_0)^2} + \frac{1}{\text{ARE}}\frac{n^2}{m}$, namely, $\widehat{n} \to \Normal{n}{\frac{n^2}{u}\frac{2P_0(1-P_0)+P_1}{(1-P_0)^2} + \frac{1}{\text{ARE}}\frac{n^2}{m}}$, where \textup{ARE} is the asymptotic relative efficiency of Procedure 1.
In addition, $P_0$ and $P_1$ can be estimated as described in Theorem \ref{thm1}.
\end{thm}

The proof is identical to that of Theorem \ref{thm12}.

\section{Simulation Results} \label{sec:sim}

In this section we validate our analysis for the asymptotic bias and variance of Algorithm \ref{alg:cardEstWithSampling} and Algorithm \ref{alg:cardEstWithSubSampling}, as stated in Theorems \ref{thm1} and \ref{thm12} respectively. We implement both algorithms using the HyperLogLog \cite{Flajolet2007} as Procedure 1, and simulate a stream of $n$ distinct elements. Each distinct element $e_j$ appears $f_j$ times in the original (unsampled) stream. These frequencies are determined according to the following models:
\begin{enumerate}
\item Uniform distribution: The frequency of the elements is uniformly distributed between $100$ and $10,000$; i.e., $f_j \sim \Unif{10^2}{10^4}$.
\item Pareto distribution: The frequency of the elements follows the heavy-tailed rule with shape parameter $\alpha$ and scale parameter $s=500$; i.e., the frequency probability function is $p(f_j) = \alpha s^\alpha f^{-\alpha-1}$, where $\alpha>0$ and $f_j \ge s >0$. The scale parameter $s$ represents the smallest possible frequency.
\end{enumerate}

Pareto distribution has several unique properties. In particular, if $\alpha \le 2$, it has infinite variance, and if $\alpha \le 1$, it has infinite mean. 
As $\alpha$ decreases, a larger portion of the probability mass is in the tail of the distribution, and it is therefore useful when a small percentage of the population controls the majority of the measured quantity.

Table~\ref{table:alg1} presents the simulation results for Algorithm \ref{alg:cardEstWithSampling} using uniformly distributed frequencies. 
The number of distinct elements is $n=10,000$. Thus, the expected length of the original stream $ X $ is $ 10,000 \cdot \frac{100 + 10,000}{2} = 50.5 \cdot 10^6$. We examine two sampling rates: $P=1/100$ (Table \ref{table:alg1}(a)) and $P=1/1000$ (Table \ref{table:alg1}(b)). We use different $m$ values, and for every $m$ average the results over $200$ different runs. 
In each table row we present, for every $m$, the bias and the variance. The bias column is only from the simulations and it is always very close to $0$, as proven in our analysis. For the variance we have two values: one from the analysis (Theorem \ref{thm1}) and one from the simulations. 

The results in Table \ref{table:alg1} show very good agreement between the simulation results and our analysis. First, as already said, the bias values are all very close to $0$. Second, the simulation variance is always very close to the analyzed variance.

\begin{table}[ht]
\centering
\begin{subfigure}[b]{0.44\textwidth}
\begin{tabular}{|c|c|c|c|}
\hline 
\multirow{2}{*}{m} & \multirow{2}{*}{bias} & \multicolumn{2}{|c|}{variance} \\ \cline{3-4}
{} & {} & analysis & simulation \\\hline 
50 & 0.0023 & 0.0200 & 0.0191 \\\hline
100 & 0.0134 & 0.0100 & 0.0116 \\\hline 
150 & 0.0094 & 0.0067 & 0.0057 \\\hline 
\end{tabular}
\caption{$P=1/100$}
\label{table:alg1A}
\end{subfigure}
\begin{subfigure}[b]{0.44\textwidth}
\begin{tabular}{|c|c|c|c|}
\hline
\multirow{2}{*}{m} & \multirow{2}{*}{bias} & \multicolumn{2}{|c|}{variance} \\ \cline{3-4}
{} & {} & analysis & simulation \\\hline 
50 & 0.0141 & 0.0209 & 0.0174 \\\hline
100 & 0.0094 & 0.0114 & 0.0099 \\\hline 
150 & 0.0036 & 0.0096 & 0.0087 \\\hline 
\end{tabular}
\caption{$P=1/1000$}
\label{table:alg1B}
\end{subfigure}
\caption{Simulation results for Algorithm \ref{alg:cardEstWithSampling} using uniformly distributed frequencies}
\label{table:alg1}
\end{table}

Next, we consider Algorithm \ref{alg:cardEstWithSubSampling} and seek to validate Theorem \ref{thm12}.
Table \ref{table:unif} presents the simulation results for uniform distribution of the frequencies. The total storage budget is $200$ units, which are partitioned between $m$ and $u$. The number of distinct elements is $n=10,000$. We examine again two sampling rates: $P=1/100$ and $P=1/1000$. Table \ref{table:par} presents results for the Pareto distribution of the frequencies, with $\alpha=1.1$, $n=10,000$, $P=1/100$, and a total storage budget of $2,000$ units. The results are averaged again over $200$ runs, and the variance from the analysis is determined according to Theorem \ref{thm12}.

\begin{table}[ht]
\centering
\begin{subfigure}[b]{0.49\textwidth}
\begin{tabular}{|c|c|c|c|c|}
\hline 
\multirow{2}{*}{m} & \multirow{2}{*}{u} & \multirow{2}{*}{bias} & \multicolumn{2}{|c|}{variance} \\ \cline{4-5}
{} & {} & {} & analysis & simulation \\\hline 
10 & 190 & 0.0439 & 0.1000 & 0.1149 \\\hline 
50 & 150 & 0.0025 & 0.0200 & 0.0217 \\\hline 
100 & 100 & 0.0029 & 0.0101 & 0.0121 \\\hline 
150 & 50 & 0.0037 & 0.0068 & 0.0075 \\\hline 
190 & 10 & 0.0058 & 0.0060 & 0.0054 \\\hline 
\end{tabular}
\caption{$P=1/100$}
\label{table:unifA}
\end{subfigure}
\begin{subfigure}[b]{0.49\textwidth}
\begin{tabular}{|c|c|c|c|c|}
\hline
\multirow{2}{*}{m} & \multirow{2}{*}{u} & \multirow{2}{*}{bias} & \multicolumn{2}{|c|}{variance} \\ \cline{4-5}
{} & {} & {} & analysis & simulation \\\hline 
10 & 190 & 0.0093 & 0.1000 & 0.1081 \\\hline 
50 & 150 & 0.0184 & 0.0200 & 0.0199 \\\hline 
100 & 100 & 0.0114 & 0.0101 & 0.0118 \\\hline 
150 & 50 & 0.0060 & 0.0068 & 0.0059 \\\hline 
190 & 10 & 0.0142 & 0.0058 & 0.0053 \\\hline 
\end{tabular}
\caption{$P=1/1000$}
\label{table:unifB}
\end{subfigure}
\caption{Simulation results for Algorithm \ref{alg:cardEstWithSubSampling} using uniform distribution and $m + u = 200$ storage units}
\label{table:unif}
\end{table}

In both tables we see again that the bias is indeed practically $0$ and that the variance of the algorithm as found by the simulations is very close to the variance found by our analysis.
These results are very consistent, for both frequency distributions, both sampling rates, and all $m$ and $u$ values.
As expected, when $m+u$ increases (more storage is used), the variance decreases.

\begin{table}[ht]
\centering 
\begin{tabular}{|c|c|c|c|c|}
\hline
\multirow{2}{*}{m} & \multirow{2}{*}{u} & \multirow{2}{*}{bias} & \multicolumn{2}{|c|}{variance} \\ \cline{4-5}
{} & {} & {} & analysis & simulation \\\hline 
50 & 1950 & 0.00005 & 0.0200 & 0.0217 \\\hline 
100 & 1900 & 0.0189 & 0.0100 & 0.0104 \\\hline
500 & 1500 & 0.0011 & 0.0020 & 0.0023 \\\hline
1000 & 1000 & 0.00001 & 0.0010 & 0.0009 \\\hline
1500 & 500 & 0.0107 & 0.0007 & 0.0006 \\\hline 
\end{tabular}
\caption{Simulation results for Algorithm \ref{alg:cardEstWithSubSampling} using Pareto distribution and $m + u = 2000$ storage units}
\label{table:par}
\end{table}

We now want to compare the performance of Algorithms \ref{alg:cardEstWithSampling} and \ref{alg:cardEstWithSubSampling}. Recall that Algorithm \ref{alg:cardEstWithSubSampling} is expected to have a higher variance, but with significantly less storage. In Theorems \ref{thm1} and \ref{thm12} we got the following closed expressions for the relative variance of the algorithms:
\begin{enumerate}
\item Algorithm \ref{alg:cardEstWithSampling}: $\frac{1}{l}\frac{P_0(1-P_0)+P_1}{(1-P_0)^2} + \frac{1}{m}$.
\item Algorithm \ref{alg:cardEstWithSubSampling}: $\frac{1}{u}\frac{2P_0(1-P_0)+P_1}{(1-P_0)^2} + \frac{1}{m}$.
\end{enumerate}
Recall that $m+l$ is the total storage used by Algorithm \ref{alg:cardEstWithSampling} ($l$ is the sample length), and $m+u$ is the total storage used by Algorithm \ref{alg:cardEstWithSubSampling}. The probabilities $P_0$ and $P_1$ are determined according to the sampling rate $P$ and the frequency distribution of the distinct elements in the stream (see Theorem \ref{thm1}). Therefore, in a given stream, the only parameters that need to be determined by the user are $m$ in Algorithm \ref{alg:cardEstWithSampling}, and $m$ and $u$ in Algorithm \ref{alg:cardEstWithSubSampling}. In order to find the values of $m$ and $u$ that yield the minimal variance for a given input stream, one only needs to know the sampling rate and then minimize the relative variance function stated above.

Table \ref{table:compAlg} presents the simulation results for $n=10,000$, a uniform distribution of element frequencies, and for several sampling rates. 
Table \ref{table:compAlg}(a) presents the variance of Algorithm \ref{alg:cardEstWithSampling}. In each table row we present the sample length $l$, the value of $m$, the total storage used by the algorithm ($m+l$), and the simulation variance (averaged over $200$ different runs). Recall that in addition to $m$, Algorithm \ref{alg:cardEstWithSampling} uses $O(l)$ storage units for the exact computation of $\abs{E_1}$.
Table \ref{table:compAlg}(b) presents the minimal variance of Algorithm \ref{alg:cardEstWithSubSampling} as a function of $B$. $B$ indicates the total number of storage units we are willing to spend. In each table row we present the optimal partition of $B$ between $m$ and $u$ that minimizes the variance of the estimator, and the simulation variance for these $m$ and $u$ values. 
For the case where $P=1$ (no sampling), we provide in both tables the simulation variance of HyperLogLog \cite{Flajolet2007}, which we use as Procedure 1. This algorithm is the best known cardinality estimator and it has a relative variance of $\Var{\frac{\widehat{n}}{n}} \approx 1.08/m$ \cite{Flajolet2007}. In this case we do not provide the values of $l$, $m$ and $u$ as there is no meaning to these parameters because sampling is not used.

\begin{table}[ht]
\footnotesize
\centering
\begin{subfigure}[b]{0.49\textwidth}
\begin{tabular}{|c|c|c|c|c|}
\hline
\multirow{2}{*}{P} & \multicolumn{3}{|c|}{storage} & variance \\ \cline{2-4}
{} & m & l & total & (simulation) \\\hline
\multirow{3}{*}{1/100} & 100 & \multirow{3}{*}{505,000} & 505,100 & 0.0116 \\ \cline{2-2}\cline{4-5}
{} & 500 & {} & 505,500 & 0.0018 \\ \cline{2-2}\cline{4-5}
{} & 1000 & {} & 506,000 & 0.0009 \\\hline
\multirow{3}{*}{1/500} & 100 & \multirow{3}{*}{101,000} & 101,100 & 0.0095 \\ \cline{2-2}\cline{4-5}
{} & 500 & {} & 101,500 & 0.0021 \\ \cline{2-2}\cline{4-5}
{} & 1000 & {} & 102,000 & 0.0008 \\\hline 
\multirow{3}{*}{1/1000} & 100 & \multirow{3}{*}{50,500} & 50,600 & 0.0099 \\ \cline{2-2}\cline{4-5}
{} & 500 & {} & 51,000 & 0.0019 \\ \cline{2-2}\cline{4-5}
{} & 1000 & {} & 51,500 & 0.0008 \\\hline
\multirow{3}{*}{1} & 100 & - & 100 & 0.0101 \\ \cline{2-5}
{} & 500 & - & 500 & 0.0021 \\ \cline{2-5}
{} & 1000 & - & 1000 & 0.0010 \\\hline
\end{tabular}
\caption{Algorithm \ref{alg:cardEstWithSampling}}
\label{table:compAlg1}
\end{subfigure}
\begin{subfigure}[b]{0.49\textwidth}
\begin{tabular}{|c|c|c|c|c|}
\hline
\multirow{2}{*}{P} & \multicolumn{3}{|c|}{storage} & variance \\ \cline{2-4}
{} & B & m & u & (simulation) \\\hline 
\multirow{3}{*}{1/100} & 100 & 92 & 8 & 0.0112 \\ \cline{2-5}
{} & 500 & 460 & 40 & 0.0022 \\ \cline{2-5}
{} & 1000 & 921 & 79 & 0.0009 \\\hline
\multirow{3}{*}{1/500} & 100 & 80 & 20 & 0.0126 \\ \cline{2-5}
{} & 500 & 401 & 99 & 0.0027 \\ \cline{2-5}
{} & 1000 & 803 & 197 & 0.0011 \\\hline
\multirow{3}{*}{1/1000} & 100 & 72 & 28 & 0.0152 \\ \cline{2-5}
{} & 500 & 363 & 137 & 0.0031 \\ \cline{2-5}
{} & 1000 & 724 & 276 & 0.0013 \\\hline
\multirow{3}{*}{1} & 100 & - & - & 0.0101 \\ \cline{2-5}
{} & 500 & - & - & 0.0021 \\ \cline{2-5}
{} & 1000 & - & - & 0.0010 \\\hline
\end{tabular}
\caption{Algorithm \ref{alg:cardEstWithSubSampling}}
\label{table:compAlg2}
\end{subfigure}
\caption{Simulation results for Algorithms \ref{alg:cardEstWithSampling} and \ref{alg:cardEstWithSubSampling} using uniform distribution}
\label{table:compAlg}
\end{table}

We can easily see from the tables that the storage-variance trade-off of
Algorithm \ref{alg:cardEstWithSubSampling} is {\em significantly better\/} than that of Algorithm \ref{alg:cardEstWithSampling}. For example, the same variance ($0.011$) is obtained by both algorithms in the first row of $P=1/100$. However, in this row Algorithm \ref{alg:cardEstWithSampling} uses 505,100 storage units whereas Algorithm \ref{alg:cardEstWithSubSampling} uses only $100$. For $P=1/500$, we see that the same variance ($0.002$) is obtained by the two algorithms when Algorithm \ref{alg:cardEstWithSampling} uses 101,500 storage units while Algorithm \ref{alg:cardEstWithSubSampling} uses only $500$.

\section{Conclusions} \label{sec:conclusion}

In this paper we studied the problem of estimating the number of distinct elements in a stream when only a small sample of the stream is given.
We presented Algorithm \ref{alg:cardEstWithSampling}, which combines a sampling process with a generic cardinality estimation procedure. The proposed algorithm consists of two steps: (a) cardinality estimation of the sampled stream using any known cardinality estimator; (b) estimation of the sampling ratio using Good-Turing frequency.
Then we presented an enhanced algorithm that uses subsampling in order to reduce the memory cost of Algorithm \ref{alg:cardEstWithSampling}.
We proved that both algorithms do not affect the asymptotic unbiasedness of the original estimator. We also analyzed the sampling effect on the asymptotic variance of the estimators. Finally, we presented simulation results that validate our analysis and showed how to find the optimal parameter values that yield the minimal variance.

\bibliographystyle{abbrv}
\bibliography{samplingPaper}

\end{document}